\documentclass[ journal ]{IEEEtran}
\usepackage[T1]{fontenc}

\ifCLASSINFOpdf
\else
\usepackage[dvips]{graphicx}
\fi
\usepackage{url}

\usepackage{amsmath}
\usepackage{amsthm}
\usepackage[cmintegrals]{newtxmath}

\usepackage{amssymb}
\usepackage{bm}
\usepackage{bbm}
\usepackage{subfigure}
\usepackage{graphicx}
\usepackage{calc}
\usepackage{epstopdf}
\newtheorem{theorem}{Theorem}

\newcommand{\norm}[1]{\left\lVert#1\right\rVert} 

\newtheorem{assumption}{Assumption}
\usepackage{verbatim}
\usepackage[ruled,vlined]{algorithm2e}
\usepackage{color}
\usepackage{array}
\usepackage{cite}
\usepackage{soul} 
\usepackage{xcolor}
\usepackage{floatrow}
\usepackage{lipsum}
\usepackage{multicol}
\usepackage{multirow} 
\usepackage{algorithmic}
 
 \usepackage{fancyhdr}
 \fancyheadoffset{0mm} 
 \fancyfootoffset{0mm} 

 \fancypagestyle{}{ 
 	\tiny
 	\fancyhf{} 
 	\rhead{\thepage}} 
\begin{document}
\title{Decentralized Consensus Optimization Based on Parallel Random Walk}

\author{
	
		Yu~Ye,~\IEEEmembership{Student~Member,~IEEE,} 
		Hao~Chen,
		Zheng~Ma,~\IEEEmembership{Member,~IEEE,}
		and~Ming~Xiao,~\IEEEmembership{Senior~Member,~IEEE,} 
	\thanks{ Yu Ye, Hao Chen, Zheng Ma and Ming Xiao are with School of Electrical Engineering and Computer Science, KTH, Stockholm, Sweden (email: yu9@kth.se, haoch@kth.se, zma@kth.se, mingx@kth.se).}	
		
}
  
\maketitle
 
\begin{abstract}
 
The alternating direction method of multipliers (ADMM) has recently been recognized as a promising approach for large-scale machine learning models. However, very few results study ADMM from the aspect of communication costs, especially jointly with running time. In this letter, we investigate the communication efficiency and running time of ADMM in solving the consensus optimization problem over decentralized networks. We first review the effort of random walk ADMM (W-ADMM), which reduces communication costs at the expense of running time. To accelerate the convergence speed of W-ADMM, we propose the parallel random walk ADMM (PW-ADMM) algorithm, where multiple random walks are active at the same time. Moreover, to further reduce the running time of PW-ADMM, the intelligent parallel random walk ADMM (IPW-ADMM) algorithm is proposed through integrating the \textit{Random Walk with Choice} with PW-ADMM. By numerical results from simulations, we demonstrate that the proposed algorithms can be both communication efficient and fast in running speed compared with state-of-the-art methods.


\end{abstract}

\begin{IEEEkeywords}
  Decentralized network; consensus optimization; alternating direction method of multipliers (ADMM).  
\end{IEEEkeywords}
\IEEEpeerreviewmaketitle

\section{Introduction}

\IEEEPARstart{C}{onsider} a network $\mathcal{G}=(\mathcal{V},\mathcal{E})$, where $\mathcal{V}=\{1,...,N\}$ is the set of agents and $\mathcal{E}$ is the set of links. The agents aim to solve the following consensus optimization problem,
\begin{equation}\label{eq1}
\min_{x}~ \sum_{i=1}^{N}f_i(x),
\end{equation}
where $f_i:\mathbbm{R}^n\to\mathbbm{R}$ is the local loss function held by agent $i$, and all the agents share a common optimization variable $x\in\mathbbm{R}^n$. This consensus problem is applied in various areas including wireless sensor networks \cite{ID,FZeng} and smart grid implementations \cite{HJ}. Specifically in \cite{5704850}, the distributed beamforming scheme for multiple relay nodes (RNs) is designed by solving a consensus problem. In general \cite{6365874}, consensus-based distributed linear estimation for cooperative communication in wireless networks can be formulated as (\ref{eq1}).

A few decentralized algorithms have been provided to solve the consensus problem in (\ref{eq1}). For low computation complexity, the first-order algorithms such as decentralized gradient decent (DGD) and EXTRA are proposed by \cite{DGD} and \cite{EXTRA}, respectively, where agents use their local gradient during the optimization process. Among existing decentralized algorithms, the alternating direction method of multipliers (ADMM), which is extensively applied in wireless communications \cite{8466658}, is shown to be faster than DGD in convergence \cite{DADMM}, in which at every iteration an agent needs to solve an optimization problem with collected information from neighboring agents. Besides, a variety of algorithms such as Gauss-Seidel ADMM and Jacobi-Proximal ADMM \cite{Deng2017} based on the original work in \cite{boyd2011distributed} are provided to solve the consensus problem (\ref{eq1}).

In practice, one ideal decentralized approach is expected to obtain the optimal solution of (\ref{eq1}) with the minimal communication and computation costs. Lots of research efforts have been put on computation complexity reduction. But there are very few results \cite{JFC,COCA,WADMM} on reducing the communication cost of ADMM. Though both proposed algorithms, distributed ADMM (D-ADMM) in \cite{JFC} and communication-censored ADMM (COCA) in \cite{COCA}, limit the overall communication at each iteration, the COCA can adaptively determine whether a message is informative, and D-ADMM relies more on the network topology.    
The random walk ADMM (W-ADMM) algorithm is proposed in \cite{WADMM}, which randomly activates a succession of nodes and incrementally updates the optimization variable. W-ADMM can achieve much less communication cost but at the expense of running time, since at each iteration only one agent is active for optimization. However, all the approaches provided in \cite{JFC,COCA,WADMM} are synchronous ADMM, which may suffer from the straggler problem \cite{zhang2014}. 

In what follows, we will propose the parallel random walk ADMM (PW-ADMM) algorithm that allows multiple random walks active in parallel. Furthermore, we integrate the intelligent agents selection scheme with PW-ADMM, which is presented in the algorithm of intelligent parallel random walk ADMM (IPW-ADMM). By numerical results, we show that the proposed approaches can be both communication efficient and fast in running time.


The remaining of this letter is organized as follows. We first introduce the parallel random walk algorithms in Section II. To demonstrate the effectiveness of the proposed approaches, we provide numerical results in Section III. Finally, we conclude the letter in Section IV. 
 
 \section{Parallel Random Walk Algorithms}  

\begin{figure*} [t] 
	\vskip 0.2in
	\begin{center}
		\centerline{\includegraphics[width=145mm]{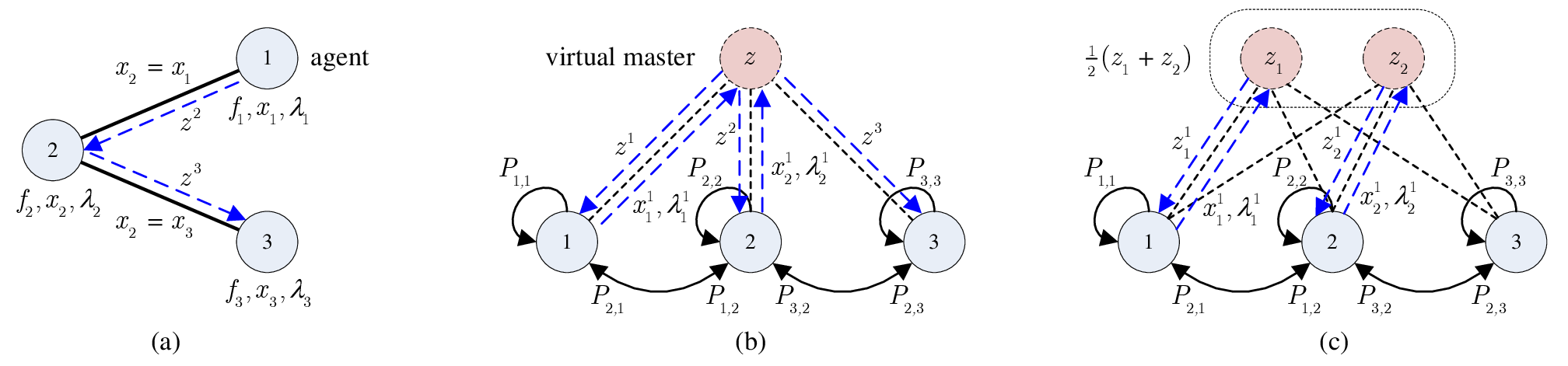}}
		\caption{(a) an example of W-ADMM; (b) the equivalent architecture of W-ADMM; (c) the equivalent architecture of parallel random walk ADMM algorithms.}
		\label{1}
	\end{center}
	\vskip -0.2in
\end{figure*} 


By defining $ \bm{x}=[x_1,...,x_N]\in\mathbbm{R}^{nN} $, problem (\ref{eq1}) can be rewritten as
\begin{equation}\label{eq2}
\begin{aligned}
\min_{\bm{x}, z}~\sum_{i=1}^{N}f_i(x_i),  ~~~
s.t.~  \mathbbm{1}\otimes z-\bm{x}=0,
\end{aligned}
\end{equation}
where $z\in\mathbbm{R}^n$, $\mathbbm{1}=[1,...,1]^T\in\mathbbm{R}^{n}$ and $\otimes$ is Kronecker product. The augmented Lagrangian for problem (\ref{eq2}) is 
\begin{equation}
\mathcal{L}_{\rho}(\bm{x},z,\bm{\lambda})=\sum _{i=1}^{N}f_i(x_i) + \left\langle \bm{\lambda} ,\mathbbm{1}\otimes z-\bm{x}   \right\rangle + \frac{\rho}{2}\norm{\mathbbm{1}\otimes z-\bm{x}}^2,
\end{equation}
where $\bm{\lambda}=[\lambda_1,...,\lambda_N]\in\mathbbm{R}^{nN}$ is the dual variable, and $\rho>0$ is a constant parameter. The iterated updates of $\bm{x}$, $\bm{\lambda}$ and $z$ can be found in Algorithm 1 (W-ADMM) in \cite{WADMM}. The Fig. 1 (a) presents an example of W-ADMM. Ignoring the difference in communication cost, the equivalent implementation of W-ADMM is shown in Fig. 1 (b), where agent $i$ updates local variables $x_i$ and $\lambda_i$ after receiving token $z$, while the virtual master updates $z$ with up-to-date $x_i$ and $\lambda_i$.  

 \subsection{Parallel Random Walk ADMM}
To introduce the parallel random walk ADMM, we extend the architecture in Fig. 1 (b) to multiple virtual masters $\mathcal{M}=\{1,...,M\}$ as Fig. 1 (c). Denoting $ \bm{z}=[z_1,...,z_M]\in\mathbbm{R}^{nM} $, the decentralized problem is given by
 \begin{equation}\label{eq3}
 \begin{aligned}
 \min_{\bm{x}, \bm{z}}~  \sum_{i=1}^{N}f_i(x_i), ~~ 
 s.t.~  \mathbbm{1}\otimes \frac{1}{M}\sum_{l=1}^{M}z_l-\bm{x}=0,
 \end{aligned}
 \end{equation} 
 where $z_l\in\mathbbm{R}^{n} ~(l\in\mathcal{M})$ is the token held by the $l$-th random walk. In the constraint of problem (4), we let $x_i$ equal to the average of the summation of tokens. By doing this, we will be able to update tokens of multiple random walks in parallel. The augmented Lagrangian for problem (\ref{eq3}) is 

 \begin{align}
 \mathcal{L}_{\rho}&(\bm{x},\bm{z},\bm{\lambda})= \sum_{i=1}^{N}f_i(x_i)\notag\\& +  \langle\bm{\lambda}, \mathbbm{1}\otimes \frac{1}{M}\sum_{l=1}^{M}z_l-\bm{x}  \rangle + \frac{\rho}{2}  \| \mathbbm{1}\otimes \frac{1}{M}\sum_{l=1}^{M}z_l-\bm{x}   \|^2.  \label{eq31}
 \end{align} 
 Following the traditional synchronous ADMM \cite{boyd2011distributed}, the update for $k+1$-th iteration follows
 \begin{subequations}
 \begin{align}
 &\frac{1}{M}\sum_{l=1}^Mz_l^{k+1} := \frac{1}{N}\sum_{i=1}^{N} (x_i^k -\frac{\lambda_i^k}{\rho}  ) ;\label{eq9} \\
 &x_i^{k+1} := \arg\min_{x_i} ~f_i(x_i) +  \langle \lambda_i^k, \frac{1}{M}\sum_{l=1}^{M}z_l^{k+1}  -x_i \rangle \notag\\&~~~~~~~~~~~~~~~~~~~ +\frac{\rho}{2} \| \frac{1}{M}\sum_{l=1}^{M}z_l^{k+1}  -x_i  \|^2,~i\in\mathcal{V}; \\
 &\lambda_i^{k+1}:=\lambda_i^k + \rho ( \frac{1}{M}\sum_{l=1}^{M}z_l^{k+1}-x_i^{k+1} ),~i\in\mathcal{V}.\label{eq11}
 \end{align}	
  \end{subequations}
Inspired by the incremental update of WADMM \cite{WADMM}, we transform (\ref{eq9})-(\ref{eq11}) to the following process by approximating $\frac{1}{M}\sum_{l=1}^{M}z_l$ with $z_m$ for the update of $x_i$ and $\lambda_i$.
\begin{subequations}
\begin{align} 
&x_i^{k+1}:=\left\{\begin{aligned}
&\arg \min_{x_i} f_i(x_i) +  \langle \lambda_i^k,z_m^{k+1}-x_i  \rangle + \frac{\rho}{2}  \| z_m^{k +1} -x_i \| ^2\\ & ~~~~~~~~~+  \frac{1}{2}\left\|x_i-x_i^{k }\right\|_{q_i}^2, ~i=i_m,~m\in\widetilde{\mathcal{M}}_{k+1} ;\\
&x_i^{k },~\text{o.w.};
\end{aligned}  \right. \label{eq13}\\
&\lambda_i^{k+1}:=\left\{\begin{aligned}
& \lambda_i^{k } + \rho\left(z_m^{k +1}-x_{i}^{k +1}\right),~i=i_m,~m\in\widetilde{\mathcal{M}}_{k+1} ;\\
& \lambda_i^{k },~\text{o.w.} ; 
\end{aligned}  \right. \label{eq14}\\
&\frac{1}{M}\sum_{l=1}^{M}z_l^{k+2} :=\frac{1}{M}\sum_{l=1}^{M}z_l^{k+1}\notag\\&~~~~~~~~~~~~~~~~~ + \sum_{i=1}^{N} [
\frac{1}{N} (x_{i}^{k +1}-\frac{\lambda_{i}^{k+1}}{\rho}  ) - \frac{1}{N} (x_i^{k }-\frac{\lambda_{i}^{k }}{\rho} )  ],\label{eq12}
\end{align}   
\end{subequations}
   \begin{algorithm}[t]
	\caption{Parallel Walk ADMM (PW-ADMM)} 
	\begin{algorithmic}[1]
		\STATE \textbf{Initialize}: $\{x_i^0=\bm{0} , \lambda_i^0=\bm{0},k_i=0|i\in\mathcal{V} \}$ and $\{z_m^1=\bm{0},k_m'=0,i^{k_m'}|i^{k_m'}\in\mathcal{V},i^{k_m'}\neq i^{k_l'},\forall l\neq m,l,m\in\mathcal{M} \}$ 
		\STATE  \setuldepth{the}\ul{\emph{\textbf{Algorithm of the $m$-th Walk:}}}   
		\FOR{$k_m'=0,1,2,...$} 
		\STATE update $x_i^{k_i+1}$ according to (\ref{pw-eq1}) with $i=i^{k'_m}$;
		\STATE update $\lambda_i^{k_i+1}$ according to (\ref{pw-eq2}) with $i=i^{k'_m}$;
		\STATE update $z_m^{k_m'+2}$ according to (\ref{pw-eq3}) with $i=i^{k'_m}$;
		\STATE set $k_i\gets k_i+1$ with $i=i^{k'_m}$; 
		\STATE choose $i^{k_m'+1}(\in\overline{\mathcal{V}}_{i^{k_m'}})$ according to $P_{i^{k_m'},i^{k_m'+1}}$;
		\STATE send $z_m^{k_m'+2}$ to node $i^{k_m'+1}$.
		\ENDFOR
	\end{algorithmic} 
\end{algorithm}  
where $\widetilde{\mathcal{M}}_{k+1}\subseteq \mathcal{M}$ is the set of active random walks at iteration $k$. We adopt proximal update for $x_i$ with $q_i=\tau_iI$, where $\tau_i>0$ is a step size penalty chosen by agent $i$ and $\|u\|_G^2=u^TGu$ is \textit{G-norm}. Note $i_m\neq i_l(l,m\in\widetilde{\mathcal{M}}_{k+1},l\neq m)$ in (\ref{eq13}) and (\ref{eq14}). Hence to make (\ref{eq12}) satisfied, the update for token $z_m$ should follow
 \begin{equation}\label{zm}
 	z_m^{k+2} := \left\{\begin{aligned}
 	&z_m^{k+1}+\frac{M}{N} (x_{i_m}^{k +1}-\frac{\lambda_{i_m}^{k +1}}{\rho}  ) - \frac{M}{N} (x_{i_m}^{k }-\frac{\lambda_{i_m}^{k }}{\rho} ), ~m\in \widetilde{\mathcal{M}}_{k+1};\\
 	&z_m^{k+1},~\text{o.w.}.
 	\end{aligned}  \right.
 \end{equation} 
The update of $z_m^{k+2}$ in (\ref{zm}) can be carried out in parallel and asynchronously since it does not require information of $z_l^{k+1}(l,m\in\widetilde{\mathcal{M}}_{k+1},l\neq m)$. Thus
 we parallelize (\ref{eq13})-(\ref{eq12}) by the following updates for $(k_m'+1)$-th step of the $m$-th random walk, 
 \begin{subequations} 
 \begin{align}
x_i^{k_i+1}:=& \arg \min_{x_i} f_i(x_i) +  \langle  \lambda_{i}^{k_i},z_m^{k_m'+1}-x_i \rangle +  \frac{\rho}{2} \| z_m^{k'_m+1}-x_i \| ^2\notag\\&+  \frac{1}{2} \|x_i-x_i^{k_i} \|_{q_i}^2, ~i=i ^{k'_m}; \label{pw-eq1}\\ \lambda_i^{k_i+1}:=&\lambda_i^{k_i} + \rho (z_m^{k'_m+1}-x_{i}^{k_i+1} ), ~i=i ^{k'_m};\label{pw-eq2} \\ 	
z_m^{k'_m+2} :=& z_m^{k'_m+1} + \frac{M}{N} (x_{i}^{k_i+1}-\frac{\lambda_{i}^{k_i+1}}{\rho}  ) - \frac{M}{N} (x_i^{k_i }-\frac{\lambda_{i}^{k_i }}{\rho} ),~i=i^{k'_m},\label{pw-eq3}
\end{align}
\end{subequations} 
where $k_m'$ is the clock held by the $m$-th random walk. In (\ref{pw-eq1}),  For agents $i\neq i^{k'_m}$, the local variables $x_i$ and $\lambda_i$ are not updated by the $m$-th random walk. 
The update of $z_m^{k'_m+2}$ only depends on $x_i^{k_i+1}$ and $\lambda_i^{k_i+1}$ from agent $i=i ^{k'_m}$ instead of (\ref{eq9}). 
Defining $\mathcal{V}_i(\subset \mathcal{V})$ the set of neighbors of agent $i(\in\mathcal{V})$ and $\overline{\mathcal{V}}_i = \mathcal{V}_i\bigcup i$, we present PW-ADMM in Algorithm 1. Similar to W-ADMM, the transition of token $z_m$ follows the embedded Markov chain with probability matrix $\bm{P}\in \mathbbm{R}^{N\times N}$. When $M=1$, the PW-ADMM reduces to W-ADMM.

 \subsection{Intelligent Parallel Random Walk ADMM}

  \begin{figure*}[ht] 
	\vskip 0.2in
	\begin{center}
		\centerline{\includegraphics[width=210mm]{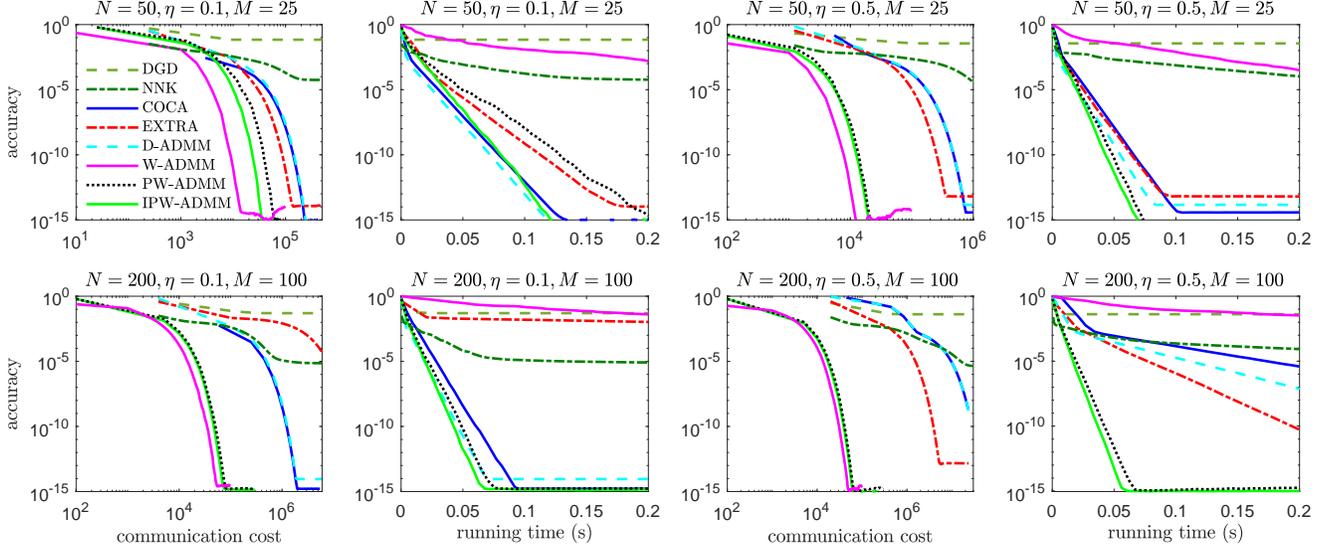}}
		\caption{The accuracy of \textit{decentralized least square} for algorithms: W-ADMM ($\beta=3$), PW-ADMM ($\rho=3,\tau=1.5$), IPW-ADMM ($\rho=3,\tau=1.5$), D-ADMM ($\rho=1$), EXTRA ($\alpha=0.05$), COCA ($c=1,\alpha=1,\rho=0.85$), DGD ($\alpha=0.01$), NNK ($K=2,\alpha=10^{-3}$).}
		\label{13}
	\end{center}
	\vskip -0.2in
\end{figure*} 

 For general problems, the convergence speeds of W-ADMM and PW-ADMM are mainly determined by how frequently all of the agents are visited. Since the transition of the token is determined by probability matrix $\bm{P}$, it is possible that the variables $x_i$ and $\lambda_i$ at some agents are updated for much fewer rounds than others. This hence may reduce the overall convergence speed. To guarantee agents not to be inactive for long time, we should improve the transition strategy for tokens. 
 Inspired by \cite{avin2006}, which introduces the \textit{Random Walk with Choice}, we present IPW-ADMM in Algorithm 2. Different from PW-ADMM, IPW-ADMM requires that agent $i$ has the knowledge of the active rounds of agents in $\overline{\mathcal{V}}_i$. Considering the $k_m'$-th step of the $m$-th random walk, the updated $z_m$ will be sent to the least visited agent $i^{k_m'+1}:=\arg\min_{i\in\mathcal{V}_{i^{k'_m}}} k_i$. Note that we do not count the communication cost of sharing $\{k_i\}$ across the agents since the cost is negligible compared with transmitting tokens. 
 
 
  \begin{algorithm}[t]
 	\caption{Intelligent PW-ADMM (IPW-ADMM)} 
 	\begin{algorithmic}[1]
 		\STATE \textbf{Initialize}: same as that of PW-ADMM.
 		\STATE  \setuldepth{the}\ul{\emph{\textbf{Algorithm of the $m$-th Walk:}}}  
 		\STATE follow steps 3-9 of PW-ADMM but substitute step 9 with the following step: 
 		\STATE choose $i^{k_m'+1}:=\arg\min_{i\in\mathcal{V}_{i^{k'_m}}}k_i $. %
 	\end{algorithmic}  
 \end{algorithm} 
 
Since all agents and parallel random walks keep individual clock, both PW-ADMM and IPW-ADMM are asynchronous algorithms. However, our proposed algorithms are different from existing work \cite{zhang2014,TChang}, where only one master updates the variable $z$. Moreover, the updated $z$ is only sent to the agents just active.

\subsection{Convergence Analysis}
 
We present some results on the convergence of PW-ADMM and IPW-ADMM with the following assumption.
 \begin{assumption}
 	The objective function $f_i(x)$ is L-Lipschitz differentiable, that is
 	\begin{equation}
 		\left\| \nabla f_i(x_1)- \nabla f_i(x_2)  \right\|\leq L\left\|x_1-x_2 \right\|,~\forall x_1,x_2\in \mathbbm{R}^n.
 	\end{equation} 
 \end{assumption}
 Though (I)PW-ADMM is asynchronous algorithm, we prove the convergence from the synchronous point of view. Without loss of generality, we denote each synchronous iteration as the update for only one token, where $|\widetilde{\mathcal{M}}_k|=1$.
 \begin{theorem}
 	Under Assumption 1 and $\tau_i=0$, the sequence $\left(\bm{x}^{k },\bm{\lambda} ^{k },\bm{z}^{k+1}\right)$ generated by (I)PW-ADMM satisfies 
 	\begin{align} 
 		&\mathcal{L}_{\rho} ( \bm{x}^{k },\bm{\lambda} ^{k },\bm{z}^{k+1}  ) - \mathcal{L}_{\rho} (\bm{x}^{k+1},\bm{\lambda} ^{k+1 },\bm{z}^{k+2}  ) \geq  (\frac{\rho}{2}-\frac{3L}{2}-\frac{L^2}{\rho} )\cdot\notag\\& \|x_{i_m}^k-x_{i_m}^{k+1} \|^2-\frac{N\rho}{2} \|\frac{1}{M}\sum_{l=1}^{M}z_l^{k+1}-z_m^{k+1}  \|^2,~\widetilde{\mathcal{M}}_{k+1}=\{m\}.\label{theorem1}
 	\end{align}
 \end{theorem}
\begin{proof}
	 By substituting $z^{k+1}$ with $\frac{1}{M}\sum_{l=1}^{M}z_l^{k+1}$ in the proof of Lemma 1 and Theorem 1 \cite{WADMM}, the result (\ref{theorem1}) can be obtained.


\end{proof}

  \begin{figure} [t] 
 	\vskip 0.2in
 	\begin{center}
 		\centerline{\includegraphics[width=95 mm]{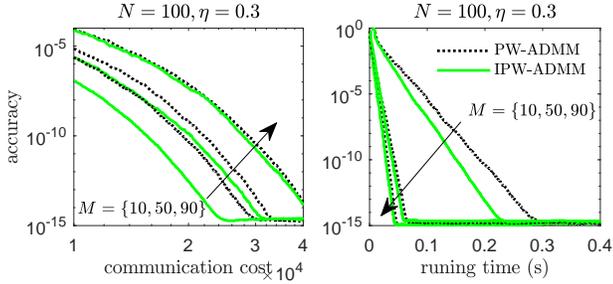}}
 		\caption{The impact of $M$ for algorithms: PW-ADMM ($\rho=1,\tau=3$), IPW-ADMM ($\rho=1,\tau=3$).}
 		\label{33}
 	\end{center}
 	\vskip -0.2in
 \end{figure}

\section{Numerical Results}
In this section, we provide numerical results from simulations to demonstrate the communication efficiency and running speed of PW-ADMM and IPW-ADMM compared with state-of-the-art methods in \cite{COCA,DGD,EXTRA,WADMM,DADMM,7590162} with respect to the accuracy, which is defined as
\begin{equation}\label{acc}
\text{accuracy}=\frac{1}{N}\sum_{i=1}^{N}\frac{\|x_i^k-x^*\|}{\|x_i^0-x^*\|},
\end{equation} 
where $x^*\in\mathbb{R}^n$ is optimal solution of (\ref{eq2}). For fair comparison, the parameters for algorithms are tuned to be the best, and kept the same in different experiments. The connected network $\mathcal{G}$ is generated randomly with $N$ agents and $|\mathcal{E}|=\frac{N(N-1)}{2}\eta$ links. Besides, the dimension of $x_i$ is set to be $n=2$. 
We consider unicast among agents, and the resultant communication cost for each transmission of a $n$-dimensional vector is $1$ unit. 
The running time includes both computing time and communication time. Without loss of generality, we assume that each agent has multi-process capability to update the tokens for multiple random walks in PW-ADMM and IPW-ADMM. Moreover, the consumed time for each communication is assumed to follow $\mathcal{U}(10^{-5},10^{-4})$ (s). The simulation is carried out on a laptop with Intel I7 processor and 8GB memory. The programing environment is Matlab R2016a.

\begin{figure*} [t] 
	\vskip 0.2in
	\begin{center}
		\centerline{\includegraphics[width=210mm]{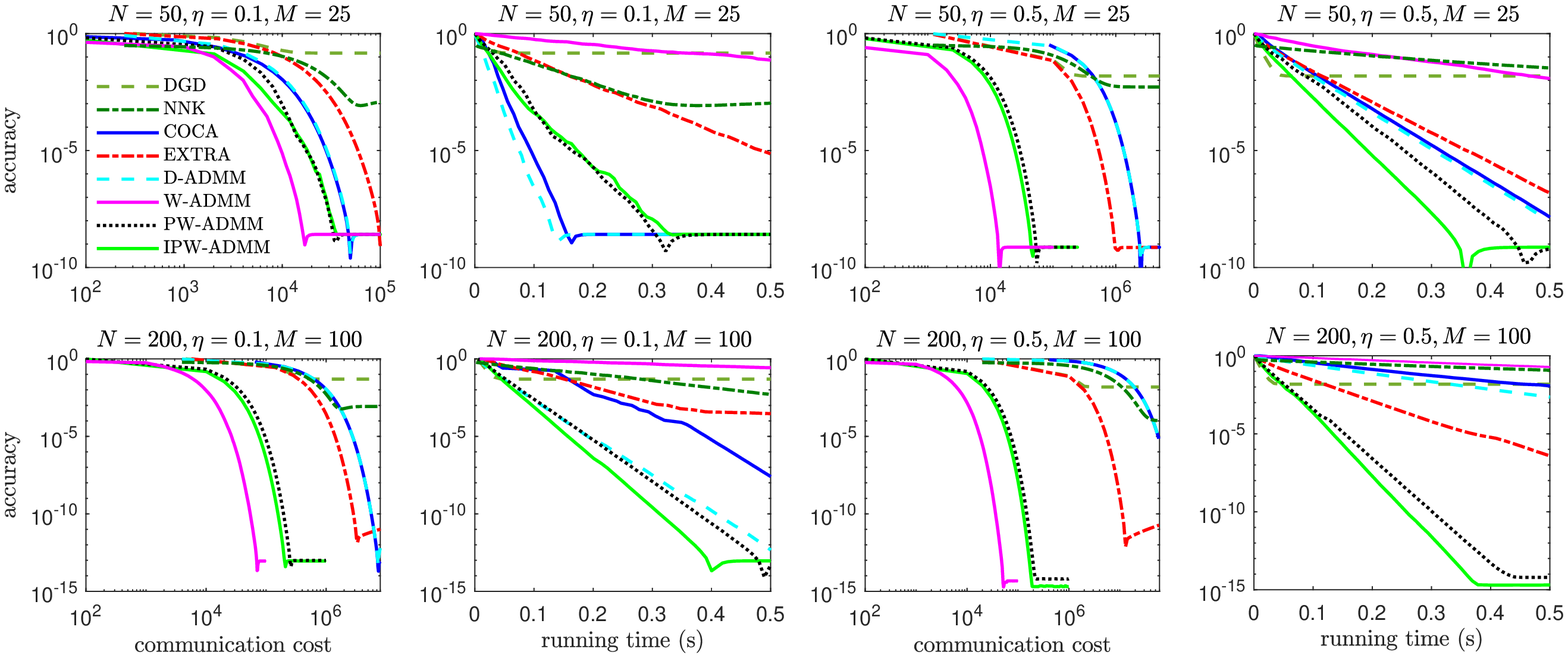}}
		\caption{The accuracy of \textit{decentralized logistic regression} for algorithms: W-ADMM ($\beta=1$), PW-ADMM ($\rho=1,\tau=3$), IPW-ADMM ($\rho=1,\tau=3$), D-ADMM ($\rho=5$), EXTRA ($c=0.01$), COCA ($c=5,\alpha=1,\rho=0.85$), DGD ($\alpha=0.01$), NNK ($K=2,\alpha=10^{-3}$).}
		\label{23}
	\end{center}
	\vskip -0.2in
\end{figure*} 

\subsection{Decentralized least square problem}
  
The decentralized least square problem such as \cite{8523816} aims at solving problem (\ref{eq1}) with the local cost function
\begin{equation}
	f_i(x_i) = \frac{1}{b_i}\sum_{j=1}^{b_i} \| x_i^To_{i,j}-t_{i,j} \|^2,
\end{equation}
where $\mathcal{D}_i=\{o_{i,j},t_{i,j}|j=1,...,b_i \}$ is the dataset of agent $i$ locally. The entries of input $o_{i,j}\in\mathbbm{R}^{2}$ and target $t_{i,j}\in\mathbbm{R}$ follow i.i.d. distribution $\mathcal{U}(0,1)$. The number of data samples is kept unique across agents with $b_i=30$.

The accuracy in (\ref{acc}) over communication cost and running time for different network settings is shown in Fig. 2. It is clear that W-ADMM is the most efficient in communication cost but with slow running speed. The proposed parallel random walk algorithms PW-ADMM and IPW-ADMM can significantly reduce the running time from W-ADMM, and consume much less communication resources compared to DGD, D-ADMM, EXTRA and COCA. In addition, IPW-ADMM can further reduce running time from PW-ADMM. Especially when the network is large and highly-connected, i.e., $N=200$ and $\eta=0.5$, the PW-ADMM and IPW-ADMM can achieve the best performance in running time with almost the same communication cost with W-ADMM. This is because the inherent asynchronous mechanism of PW-ADMM and IPW-ADMM outperforms the synchronous methods.
 
 In Fig. 3 we present the impact of $M$, the number of active random walks, on the convergence behavior. It can be concluded that with increasing $M$, a larger communication cost of both PW-ADMM and IPW-ADMM is required to achieve the same accuracy, while the running time will be shorten. Hence there exists a trade-off between the communication cost and the running time over $M$. Besides, for a larger $M$, e.g. $M=90$, the accuracy gap between PW-ADMM and IPW-ADMM shrinks compared with the case where $M=10$. This shows the advantage of intelligently choosing updating path for each walk according to the updated frequency of agents over randomly processing is weakened when more random walks are active.

\subsection{Decentralized logistic regression problem}
In the decentralized logistic regression, the local loss function of agent $i$ is
\begin{equation}
	f_i(x_i) = \frac{1}{b_i}\sum_{j=1}^{b_i}  \log( 1+\exp(-t_{i,j}x_i^To_{i,j})) ,
\end{equation}
 where $t_{i,j}\in\{-1,1\}$ and $b_i=30$. Each sample feature $o_{i,j}$ follows $\mathcal{N}(0,I)$. To generate $t_{i,j}$, we first generate a random vector $x_0\in\mathbbm{R}^2\sim\mathcal{N}(0,I)$. Then for each sample, we generate $v_{i,j}$ according to $\mathcal{U}(0,1)$, and if $v_{i,j}\leq (1+\exp(-x_0^To_{i,j}))^{-1}$, we set $t_{i,j}$ as $1$, otherwise $-1$. Since it is difficult to solve the optimization problem, e.g. (\ref{pw-eq1}), in PW-ADMM, we alternatively use the first-order approximation as
 \begin{equation}
 	f_i(x_i) \approx f_i(x_i^k)+ \nabla f_i(x_i^k) (x_i-x_i^k ).
 \end{equation} 
Fairly, we adopt the first-order approximation for algorithms IPW-ADMM, W-ADMM, D-ADMM and COCA. 
Fig. 4 presents the accuracy over communication cost and running time. Apparently compared to other benchmarks in \cite{COCA,DADMM,DGD,EXTRA,WADMM}, only the proposed parallel random walk algorithms PW-ADMM and IPW-ADMM can guarantee both the communication-efficiency and fast convergence speed for different network settings. The curves with different network setups present the similar trends as those of Fig. 2. 
 
\section{Conclusions}

We study the communication efficiency and running time for ADMM based consensus problem in decentralized networks. By extending W-ADMM, two parallel random walk algorithms, PW-ADMM and IPW-ADMM are proposed. With the asynchronous characteristic, compared with other approaches, the provided algorithms can achieve much faster running speed with less communication costs, especially for the dense networks. Moreover, simulations demonstrate the scalability of PW-ADMM and IPW-ADMM in terms of the network size.

%
 

\bibliography{ref}

\begin{thebibliography}{10}
\providecommand{\url}[1]{#1}
\csname url@samestyle\endcsname
\providecommand{\newblock}{\relax}
\providecommand{\bibinfo}[2]{#2}
\providecommand{\BIBentrySTDinterwordspacing}{\spaceskip=0pt\relax}
\providecommand{\BIBentryALTinterwordstretchfactor}{4}
\providecommand{\BIBentryALTinterwordspacing}{\spaceskip=\fontdimen2\font plus
\BIBentryALTinterwordstretchfactor\fontdimen3\font minus
  \fontdimen4\font\relax}
\providecommand{\BIBforeignlanguage}[2]{{%
\expandafter\ifx\csname l@#1\endcsname\relax
\typeout{** WARNING: IEEEtran.bst: No hyphenation pattern has been}%
\typeout{** loaded for the language `#1'. Using the pattern for}%
\typeout{** the default language instead.}%
\else
\language=\csname l@#1\endcsname
\fi
#2}}
\providecommand{\BIBdecl}{\relax}
\BIBdecl

\bibitem{ID}
I.~D. {Schizas}, A.~{Ribeiro}, and G.~B. {Giannakis}, ``Consensus in ad-hocwsns
  with noisy links-part i: Distributed estimation of deterministic signals,''
  \emph{IEEE Transactions on Signal Processing}, vol.~56, no.~1, pp. 350--364,
  Jan 2008.

\bibitem{FZeng}
F.~{Zeng}, C.~{Li}, and Z.~{Tian}, ``Distributed compressive spectrum sensing
  in cooperative multihop cognitive networks,'' \emph{IEEE Journal of Selected
  Topics in Signal Processing}, vol.~5, no.~1, pp. 37--48, Feb 2011.

\bibitem{HJ}
H.~J. {Liu}, W.~{Shi}, and H.~{Zhu}, ``Distributed voltage control in
  distribution networks: Online and robust implementations,'' \emph{IEEE
  Transactions on Smart Grid}, vol.~9, no.~6, pp. 6106--6117, Nov 2018.

\bibitem{5704850}
J.~{Choi}, ``Distributed beamforming using a consensus algorithm for
  cooperative relay networks,'' \emph{IEEE Communications Letters}, vol.~15,
  no.~4, pp. 368--370, April 2011.

\bibitem{6365874}
H.~{Paul}, J.~{Fliege}, and A.~{Dekorsy}, ``In-network-processing: Distributed
  consensus-based linear estimation,'' \emph{IEEE Communications Letters},
  vol.~17, no.~1, pp. 59--62, January 2013.

\bibitem{DGD}
K.~Yuan, Q.~Ling, and W.~Yin, ``On the convergence of decentralized gradient
  descent,'' \emph{SIAM Journal on Optimization}, vol.~26, no.~3, pp.
  1835--1854, 2016.

\bibitem{EXTRA}
W.~Shi, Q.~Ling, G.~Wu, and W.~Yin, ``Extra: An exact first-order algorithm for
  decentralized consensus optimization,'' \emph{SIAM Journal on Optimization},
  vol.~25, no.~2, pp. 944--966, 2015.

\bibitem{8466658}
E.~{Vlachos}, G.~C. {Alexandropoulos}, and J.~{Thompson}, ``Massive mimo
  channel estimation for millimeter wave systems via matrix completion,''
  \emph{IEEE Signal Processing Letters}, vol.~25, no.~11, pp. 1675--1679, Nov
  2018.

\bibitem{DADMM}
W.~{Shi}, Q.~{Ling}, K.~{Yuan}, G.~{Wu}, and W.~{Yin}, ``On the linear
  convergence of the admm in decentralized consensus optimization,'' \emph{IEEE
  Transactions on Signal Processing}, vol.~62, no.~7, pp. 1750--1761, April
  2014.

\bibitem{Deng2017}
W.~Deng, M.-J. Lai, Z.~Peng, and W.~Yin, ``Parallel multi-block admm with $o(1
  / k)$ convergence,'' \emph{Journal of Scientific Computing}, vol.~71, no.~2,
  pp. 712--736, May 2017.

\bibitem{boyd2011distributed}
S.~Boyd, N.~Parikh, E.~Chu, B.~Peleato, J.~Eckstein \emph{et~al.},
  ``Distributed optimization and statistical learning via the alternating
  direction method of multipliers,'' \emph{Foundations and Trends in Machine
  learning}, vol.~3, no.~1, pp. 1--122, 2011.

\bibitem{JFC}
J.~F.~C. {Mota}, J.~M.~F. {Xavier}, P.~M.~Q. {Aguiar}, and M.~{Puschel},
  ``D-admm: A communication-efficient distributed algorithm for separable
  optimization,'' \emph{IEEE Transactions on Signal Processing}, vol.~61,
  no.~10, pp. 2718--2723, May 2013.

\bibitem{COCA}
Y.~Liu, W.~Xu, G.~Wu, Z.~Tian, and Q.~Ling, ``Communication-censored admm for
  decentralized consensus optimization,'' \emph{IEEE Transactions on Signal
  Processing}, vol.~67, no.~10, pp. 2565--2579, 2019.

\bibitem{WADMM}
W.~Yin, X.~Mao, K.~Yuan, Y.~Gu, and A.~H. Sayed, ``A communication-efficient
  random-walk algorithm for decentralized optimization,'' \emph{arXiv preprint
  arXiv:1804.06568}, 2018.

\bibitem{zhang2014}
R.~Zhang and J.~Kwok, ``Asynchronous distributed admm for consensus
  optimization,'' in \emph{International Conference on Machine Learning}, 2014,
  pp. 1701--1709.

\bibitem{avin2006}
C.~Avin and B.~Krishnamachari, ``The power of choice in random walks: An
  empirical study,'' in \emph{Proceedings of the 9th ACM international
  symposium on Modeling analysis and simulation of wireless and mobile
  systems}.\hskip 1em plus 0.5em minus 0.4em\relax ACM, 2006, pp. 219--228.

\bibitem{TChang}
T.~{Chang}, M.~{Hong}, W.~{Liao}, and X.~{Wang}, ``Asynchronous distributed
  admm for large-scale optimization-part i: Algorithm and convergence
  analysis,'' \emph{IEEE Transactions on Signal Processing}, vol.~64, no.~12,
  pp. 3118--3130, June 2016.

\bibitem{7590162}
A.~{Mokhtari}, Q.~{Ling}, and A.~{Ribeiro}, ``Network newton distributed
  optimization methods,'' \emph{IEEE Transactions on Signal Processing},
  vol.~65, no.~1, pp. 146--161, Jan 2017.

\bibitem{8523816}
C.~{Huang}, L.~{Liu}, C.~{Yuen}, and S.~{Sun}, ``Iterative channel estimation
  using lse and sparse message passing for mmwave mimo systems,'' \emph{IEEE
  Transactions on Signal Processing}, vol.~67, no.~1, pp. 245--259, Jan 2019.

\end{thebibliography}
\bibliographystyle{IEEEtran}

\end{document}